\documentclass{article}

\usepackage[utf8]{inputenc}
\usepackage[english]{babel}
\usepackage{todonotes}

\advance\textheight by 2cm
\advance\topmargin -1cm

\advance\textwidth by 2cm
\advance\oddsidemargin by -1cm
\advance\evensidemargin by -1cm

\usepackage{graphicx}
\usepackage{centernot}
\usepackage{amssymb}
\usepackage{bm}
\usepackage{mathtools}
\usepackage{multirow}
\usepackage{tikz}
\usetikzlibrary{calc,positioning}

\newtheorem{theorem}{Theorem}[section]
\newtheorem{corollary}{Corollary}[section]
\newtheorem{lemma}{Lemma}[section]
\newtheorem{definition}{Definition}[section]

\newcounter{noqed}
\newcommand{\qed}{ \ifmmode\mbox{ }\fi\rule[-.05em]{.3em}{.7em}\setcounter{noqed}{0}}
\newenvironment{proof}[1][{}]{\noindent{\bf Proof#1. }\setcounter{noqed}{1}}{\ifnum\value{noqed}=1\qed\fi\par\medskip}
\newenvironment{renumerate}{\begin{enumerate}}{\end{enumerate}}

\renewcommand{\epsilon}{\varepsilon}

\def\adj{\mathrel-\joinrel\joinrel\mathrel-}   %long minus (Knuth)
\def\noadj{\hskip3pt\not\hskip-3pt\adj}  %not long minus (Knuth)
\def\..{\,\mathpunct{\ldotp\ldotp}} % intervals
\newcommand{\nott}[1]{\bar{#1}}

\title{Semi-Monotonicity for Spectral Centrality Measures}
\author{\textbf{Paolo Boldi \;\; Davide D'Ascenzo \;\; Flavio Furia \;\; Sebastiano Vigna}}
\date{%
    Dipartimento di Informatica, Università degli Studi di Milano, Italy\\[2ex]%
    \today
}

\begin{document}
\maketitle

\begin{abstract}
	\emph{Score monotonicity} and \emph{rank monotonicity} are properties describing the behavior of
	a centrality measure when an arc is added to a network: the former requires that the score of
	the target of the arc should increase, the latter that its importance with respect to the
	remaining nodes should not deteriorate. While in directed networks almost all classical
	centrality measures satisfy both properties, in undirected networks they fail for most measures:
	adding an edge can \emph{reduce} the score or the rank of one of its endpoints.
	\emph{Semi-monotonicity} is a recently introduced weaker property for undirected networks,
	requiring that \emph{at least one} of the two endpoints of the new edge enjoys monotonicity, and
	it is known to hold for closeness, harmonic centrality, distance-decay centralities and
	betweenness. In this paper we study semi-monotonicity for three classical \emph{spectral}
	centrality measures: eigenvector centrality, Katz's index, and PageRank. We
	prove that all of them are \emph{strictly} rank semi-monotone on connected undirected graphs,
	and that all of them are also score semi-monotone, (the only exception being eigenvector
	centrality when scores are normalized by projecting the constant vector on the dominant
	eigenspace, for which we provide a counterexample). In particular, adding an edge can never
	decrease the PageRank of both its endpoints, which answers a question left open in previous
	work.
\end{abstract}

\section{Introduction}
\label{sec:intro}

Centrality measures are a fundamental tool for the analysis of complex networks: they are used to
identify the most important nodes of a graph, and they have been applied in a wide range of
contexts, from social network analysis to the study of the World Wide Web. To understand to which
extent a centrality measure captures the intuitive notion of importance, it is useful to state
formally properties (a.k.a.~axioms) that a centrality should (or should not) satisfy and to check
which measures satisfy them~\cite{axioms}.

Two such properties describe the behavior of a centrality when the network is modified by the
addition of an arc~\cite{axioms,BLVRMCM,BoVRMCMC}: \emph{score monotonicity} requires that the score
of the target of the new arc increases; \emph{rank monotonicity} requires that the target of the arc
does not lose ground with respect to the other nodes: every node that used to be less important than
the target is still less important after the addition. In directed networks almost all classical
centrality measures satisfy both properties~\cite{BLVRMCM,BoVRMCMC}. In undirected networks, if we
insist that both endpoints of the new edge should enjoy the increase in score and rank, the
situation is strikingly different~\cite{boldi_furia_vigna_2023}: closeness, harmonic centrality,
betweenness, eigenvector centrality, Katz's index and PageRank are not rank monotone anymore;
betweenness and PageRank are \emph{not even score monotone}. In terms of social networks: getting a
new follower is always a good thing, but getting a new friend might be detrimental to one's
importance.

These negative results motivated the introduction of \emph{semi-monotonicity}~\cite{BDFSRSM}: a
centrality is score (rank) semi-monotone if, when an edge is added to an undirected network,
\emph{at least one} of its two endpoints enjoys an increase in score (rank). It was proved
in~\cite{BDFSRSM} that closeness, harmonic centrality, a large class of distance-decay centralities
and betweenness are rank semi-monotone; the proofs are based on a stronger, additive property of
shortest-path distances called \emph{basin dominance}, which implies
\emph{$\delta$-semi-monotonicity}: the score increase at one endpoint of the new edge is at least as
large as the score increase at every other vertex. The question of whether spectral centralities are
semi-monotone was left open in~\cite{BDFSRSM}.

In this paper we study semi-monotonicity for the classical \emph{spectral} centrality
measures~\cite{vigna}: eigenvector centrality~\cite{landau,berman}, Katz's index~\cite{katz}, and
PageRank~\cite{pagerank}. We do not discuss the fourth classical spectral measure, Seeley's
index~\cite{seeley}: on connected undirected graphs it is just $\ell_1$-normalized degree, which is
score monotone and strictly rank monotone~\cite{boldi_furia_vigna_2023}, and thus, a fortiori, score
semi-monotone and strictly rank semi-monotone. Spectral measures are not defined in terms
of shortest paths, so basin dominance is of no use, and we will see that even
$\delta$-semi-monotonicity fails for them. Nonetheless, we prove that \emph{all three measures are
strictly rank semi-monotone} on connected undirected graphs, and that they are also score
semi-monotone, with a single exception: eigenvector centrality is score semi-monotone when scores
are normalized in $\ell_1$ or $\ell_2$ norm (indeed, for every monotone normalization), but not when
they are normalized by projecting the constant vector on the dominant eigenspace, a normalization
that was also considered in~\cite{boldi_furia_vigna_2023}.

As a consequence, adding an edge can never decrease the PageRank of \emph{both} its endpoints, which
answers a question left open in~\cite{boldi_furia_vigna_2023}. Connectivity is essential for
eigenvector centrality, but it is not needed for Katz's index, and for PageRank it can be replaced
by the weaker requirement that all
scores be nonzero, which is the hypothesis of the rank-monotonicity results
of~\cite{boldi_furia_vigna_2023}.

All positive results are obtained through a new property, \emph{ratio semi-monotonicity}, which is
the multiplicative analogue of $\delta$-semi-monotonicity: the \emph{relative} score increase
$c'(z)/c(z)$ caused by the new edge is maximized at one of its endpoints. Ratio semi-monotonicity
implies rank semi-monotonicity by a one-line argument, and it turns out to be the natural property
for spectral measures: for each of them, the score of a vertex $z$ different from the endpoints $x$
and $y$ of the new edge can be written as a positive combination of the scores of $x$ and $y$, plus
a positive term accounting for the walks from $z$ that never reach $x$ or $y$. This decomposition,
obtained by splitting walks at their first passage through the edge $\{x,y\}$, combined with the
elementary inequality about mediants $(p_1+p_2)/(q_1+q_2)\leq\max(p_1/q_1,p_2/q_2)$, yields all our
results with short and elementary proofs. For PageRank, we also obtain a precise description of
\emph{which} endpoint gains score, in terms of the PageRank flow along the new edge.
Following~\cite{BLVRMCM}, where strict rank monotonicity in the directed case is proved for
arbitrary damped spectral centralities~\cite{vigna} by considering updates of a single row of the
underlying matrix, we obtain all positive results as instances of a comparison theorem for two
damped spectral centralities with the same matrix: for Katz's index and PageRank the theorem is applied
to the increment caused by the new edge, which is itself a damped spectral centrality, whereas for
eigenvector centrality it is applied to the new and the old eigenvector, both of which are damped
spectral centralities of the old graph.

The paper is organized as follows. In Section~\ref{sec:defs} we recall the definitions of score and
rank (semi-)monotonicity, and introduce ratio semi-monotonicity, proving that it implies rank
semi-monotonicity. In Section~\ref{sec:walks} we describe the first-passage and last-passage
decompositions of walks used in all proofs. Section~\ref{sec:damped} contains the comparison theorem for damped spectral centralities and its
consequences for updates; Sections~\ref{sec:katz}, \ref{sec:eigenvector} and~\ref{sec:pagerank}
contain the results about Katz's index, eigenvector centrality and PageRank; along the way we show
that Katz's index is not $\delta$-semi-monotone, and that eigenvector centrality is not score
semi-monotone under the projection normalization. Section~\ref{sec:loops} discusses the impact
of loops on our results, and Section~\ref{sec:conclusions} draws some conclusions.

\section{Score, Rank and Ratio Semi-Monotonicity}
\label{sec:defs}

An \emph{undirected graph} $G=(V_G,E_G)$ is given by a set of vertices $V_G$ and a set $E_G$ of
unordered pairs of distinct vertices, called \emph{edges}; we write $x\adj y$ if there is an edge
between $x$ and $y$, and say that $x$ and $y$ are \emph{adjacent}; $x\noadj y$ is the negation of
$x\adj y$. We denote with $G + x\adj y$ the graph obtained by adding the edge $x\adj y$ to $G$. The
\emph{degree} $d(x)$ of $x$ is the number of vertices adjacent to $x$. A \emph{walk} of length $k$
from $x$ to $y$ is a sequence of vertices $x=u_0\adj u_1\adj\cdots\adj u_k=y$; the empty walk from
$x$ to $x$ has length zero. A \emph{centrality} is a function $c$ assigning a real score $c(x)$ to
every vertex $x$ of every graph in a given class of graphs; in this paper we consider connected
undirected graphs.

Throughout the paper, we will be in a situation where $x$ and $y$ are distinct, non-adjacent
vertices of an undirected graph $G$: we will then uniformly denote with $c$ the value of a
centrality on $G$, with $c'$ its value on $G'=G+x\adj y$, and use the same convention for all other
quantities (e.g., $d'(z)$ is the degree of $z$ in $G'$).

We start by recalling the definitions of score and rank monotonicity on undirected
graphs~\cite{boldi_furia_vigna_2023}, which extend the corresponding notions for directed
graphs~\cite{axioms,BLVRMCM,BoVRMCMC}:

\begin{definition}[Score monotonicity]
	\label{def:score-mon}
	Given an undirected graph $G$, a centrality $c$ is said to be \emph{score monotone on $G$} iff
	for every pair of distinct, non-adjacent vertices $x$ and $y$ we have that
	\[
		c(x) < c'(x)\quad\text{and}\quad c(y) < c'(y),
	\]
	where $c'$ is the value of the centrality on the graph $G+x\adj y$.
\end{definition}

\begin{definition}[Rank monotonicity]
	\label{def:rank-mon}
	Given an undirected graph $G$, a centrality $c$ is said to be \emph{rank monotone on $G$} iff
	for every pair of distinct, non-adjacent vertices $x$ and $y$ and for all vertices $z\neq x,y$:
	\begin{align*}
		 & c(z) < c( x) \Rightarrow c'(z) < c'(x) \text{ and } c(z) = c(x) \Rightarrow c'(z) \leq c'(x); \\
		 & c(z) < c(y) \Rightarrow c'(z) < c'(y) \text{ and } c(z) = c(y) \Rightarrow c'(z) \leq c'(y).
	\end{align*}
	It is \emph{strictly rank monotone on $G$} iff for all $z\neq x,y$:
	\[
		c(z) \leq c(x) \Rightarrow c'(z) < c'(x)\qquad\text{and}\qquad c(z) \leq c(y) \Rightarrow c'(z) < c'(y).
	\]
	In both cases, $c'$ is the value of the centrality on the graph $G+x\adj y$.
\end{definition}

In general, we will say that a property holds \emph{on a set of graphs} if it holds on all the
graphs from the set. Semi-monotonicity~\cite{BDFSRSM} requires that just one of the two endpoints of
the new edge enjoys the corresponding monotonicity property:

\begin{definition}[Score semi-monotonicity]
	\label{def:score-semi}
	Given an undirected graph $G$, a centrality $c$ is said to be \emph{score semi-monotone on $G$}
	iff for every pair of distinct, non-adjacent vertices $x$ and $y$ we have that
	\[
		c(x) < c'(x)\quad\text{or}\quad c(y) < c'(y),
	\]
	where $c'$ is the value of the centrality on the graph $G+x\adj y$.
\end{definition}

\begin{definition}[Rank semi-monotonicity]
	\label{def:rank-semi}
	Given an undirected graph $G$, a centrality $c$ is said to be \emph{rank semi-monotone on $G$}
	iff for every pair of distinct, non-adjacent vertices $x$ and $y$ at least one of the following
	two statements holds:
	\begin{itemize}
		\item for all vertices $z\neq x,y$: $c(z) < c( x) \Rightarrow c'(z) < c'(x)$ and $c(z) =
		      c(x) \Rightarrow c'(z) \leq c'(x)$;
		\item for all vertices $z\neq x,y$: $c(z) < c( y) \Rightarrow c'(z) < c'(y)$ and $c(z) =
		      c(y) \Rightarrow c'(z) \leq c'(y)$.
	\end{itemize}
	It is \emph{strictly rank semi-monotone on $G$} iff at least one of the following two statements
	holds:
	\begin{itemize}
		\item for all vertices $z\neq x,y$: $c(z) \leq c(x) \Rightarrow c'(z) < c'(x)$;
		\item for all vertices $z\neq x,y$: $c(z) \leq c(y) \Rightarrow c'(z) < c'(y)$.
	\end{itemize}
	In both cases, $c'$ is the value of the centrality on the graph $G+x\adj y$.
\end{definition}

The rank semi-monotonicity results of~\cite{BDFSRSM} are all obtained through the following
sufficient condition, which is stated in terms of score differences rather than in terms of order:

\begin{definition}[$\delta$-semi-monotonicity~\cite{BDFSRSM}]
	\label{def:delta-semi}
	Given an undirected graph $G$, a centrality $c$ is said to be \emph{$\delta$-semi-monotone on
	$G$} iff for every pair of distinct, non-adjacent vertices $x$ and $y$
	\begin{align*}
		c'(z) - c(z) \leq c'(x) - c(x) \qquad & \text{for every $z\neq x,y$} \qquad \text{or} \\
		c'(z) - c(z) \leq c'(y) - c(y) \qquad & \text{for every $z\neq x,y$},
	\end{align*}
	where $c'$ is the value of the centrality on the graph $G+x\adj y$. It is said to be
	\emph{strictly $\delta$-semi-monotone} if the inequalities are strict.
\end{definition}

It is immediate that (strict) $\delta$-semi-monotonicity implies (strict) rank
semi-monotonicity~\cite{BDFSRSM}. In this paper we introduce a multiplicative version of the same
idea. From now on, all centralities are assumed to be strictly positive on the graphs under
consideration (this is the case for all spectral centralities on connected graphs).

\begin{definition}[Ratio semi-monotonicity]
	\label{def:ratio}
	Given an undirected graph $G$, a positive centrality $c$ is said to be \emph{ratio semi-monotone
	on $G$} iff for every pair of distinct, non-adjacent vertices $x$ and $y$
	\begin{align*}
		\frac{c'(z)}{c(z)} \leq \frac{c'(x)}{c(x)} \qquad & \text{for every $z\neq x,y$} \qquad \text{or}
		\\
		\frac{c'(z)}{c(z)} \leq \frac{c'(y)}{c(y)} \qquad & \text{for every $z\neq x,y$},
	\end{align*}
	where $c'$ is the value of the centrality on the graph $G+x\adj y$. It is said to be
	\emph{strictly ratio semi-monotone} if the inequalities are strict.
\end{definition}

In other words, the \emph{relative} score increase must be maximized at one of the two endpoints of
the new edge. Note that neither $\delta$-semi-monotonicity implies ratio semi-monotonicity nor vice
versa. However, ratio semi-monotonicity is a sufficient condition for rank semi-monotonicity,
exactly like $\delta$-semi-monotonicity:

\begin{theorem}
	\label{thm:ratio-semi-mon}
	If a positive centrality measure is (strictly) ratio semi-monotone on a graph then it is
	(strictly) rank semi-monotone on the same graph.
\end{theorem}
\begin{proof}
	Let $e\in\{x,y\}$ be an endpoint for which the condition of Definition~\ref{def:ratio} holds,
	and let $z\neq x,y$ be such that $c(z)\leq c(e)$. Then
	\[
		c'(z) = c(z)\,\frac{c'(z)}{c(z)} \leq c(z)\,\frac{c'(e)}{c(e)}\leq c(e)\,\frac{c'(e)}{c(e)} = c'(e),
	\]
	and the first inequality is strict in the strict case, whereas the second inequality is strict
	if $c(z)<c(e)$. Thus, $c(z)<c(e)$ implies $c'(z)<c'(e)$ and $c(z)= c(e)$ implies $c'(z)\leq
	c'(e)$; in the strict case $c(z)\leq c(e)$ implies $c'(z)<c'(e)$.
\end{proof}

Note that a strictly ratio semi-monotone centrality is score semi-monotone as soon as the sum of the
scores of all vertices (or, more generally, any monotone norm of the score vector) does not decrease
when an edge is added: if both endpoints did not gain score, all other ratios would be smaller than
one, and the sum would strictly decrease. We will use this observation several times.

All proofs of ratio semi-monotonicity in this paper will use the following elementary inequality
about \emph{mediants}:
\begin{lemma}
	\label{lemma:mediant}
	Let $p_1$, $p_2$ be real numbers and $q_1$, $q_2$ be nonnegative real numbers that are not both
	zero, and assume that $p_i=0$ whenever $q_i=0$. Then
	\[
		\frac{p_1+p_2}{q_1+q_2}\leq \max_{i\,:\,q_i>0}\frac{p_i}{q_i}.
	\]
\end{lemma}
\begin{proof}
	Let $m$ be the maximum on the right-hand side. Then $p_i\leq mq_i$ for every $i$ (trivially if
	$q_i=0$), and summing we obtain $p_1+p_2\leq m(q_1+q_2)$.
\end{proof}

\section{First-Passage Decomposition of Walks}
\label{sec:walks}

All spectral centralities we consider can be expressed in terms of sums of weights of walks. Let $W$
be a nonnegative matrix indexed by the vertices of $G$. A \emph{walk of $W$} of length $k$ from $u$
to $v$ is a sequence of vertices $u=u_0,u_1,\dots,u_k=v$ such that $W_{u_{i}u_{i+1}}>0$ for $0\le
i<k$: let us define its \emph{weight} as the product $W_{u_0u_1}W_{u_1u_2}\cdots W_{u_{k-1}u_k}$
(the empty walk having weight one). In all cases except for Section~\ref{sec:damped}, $W_{uv}$ is
nonzero only if $u\adj v$, and the walks of $W$ are exactly the walks of $G$: we will then simply
speak of walks. If we assume that the spectral radius of $W$ is smaller than 1, so that
\[
	N=(1-W)^{-1}=\sum_{k\geq 0}W^k
\]
is well defined and nonnegative, then $N_{uv}$ is the sum of the weights of all walks of $W$ from
$u$ to $v$.

Now fix a set of vertices $T$. For $u\in V_G$ and $z\in T$, let $\Phi_{uz}$ be the sum of the
weights of all walks from $u$ to $z$ that touch $T$ only at their last vertex (a
\emph{first-passage} walk to $z$); in particular, if $u\in T$ then $\Phi_{uz}$ is one if $u=z$ and
zero otherwise. Symmetrically, for $z\in T$ and $v\in V_G$, let $\Psi_{zv}$ be the sum of the
weights of all walks from $z$ to $v$ that touch $T$ only at their first vertex (a
\emph{last-passage} walk from $z$); if $v\in T$ then $\Psi_{zv}$ is one if $v=z$ and zero otherwise.
Finally, let $\nott N_{uv}$ be the sum of the weights of all walks from $u$ to $v$ that do not touch
$T$ at all (so $\nott N_{uv}=0$ if $u\in T$ or $v\in T$, and $\nott N_{uu}\geq 1$ if $u\notin T$,
because of the empty walk). Splitting each walk at its first, or at its last, vertex in $T$, if any,
we obtain:
\begin{lemma}[First- and last-passage decomposition]
	\label{lemma:first-passage}
	For all vertices $u$, $v$,
	\[
		N_{uv}=\nott N_{uv}+\sum_{z\in T}\Phi_{uz}N_{zv}
		\qquad\text{and}\qquad
		N_{uv}=\nott N_{uv}+\sum_{z\in T}N_{uz}\Psi_{zv}.
	\]
\end{lemma}
\begin{proof}
	A walk from $u$ to $v$ either does not touch $T$, or it has a first vertex $z\in T$; in the
	latter case, it splits uniquely into a first-passage walk from $u$ to $z$ followed by a walk
	from $z$ to $v$, and the weight of the walk is the product of the weights of the two parts.
	Conversely, every such pair of walks gives a walk from $u$ to $v$ whose first vertex in $T$ is
	$z$. This proves the first identity; the second one is proved in the same way, splitting the
	walk at its last vertex in $T$ (equivalently, it is the first identity for the transpose of
	$W$).
\end{proof}
Note that $\Psi_{zv}$ is the transpose-of-$W$ analogue of $\Phi_{vz}$: if $W$ is symmetric,
$\Psi_{zv}=\Phi_{vz}$.

Note that if $W$ is row-substochastic, that is, $W=\alpha P$ for a row-stochastic matrix $P$ and
$0<\alpha<1$, all quantities above have a probabilistic meaning in terms of the Markov chain with
transition matrix $P$ \emph{killed} with probability $1-\alpha$ at each step: $N_{uv}$ is the
expected number of visits to $v$ (counting the starting point) of the killed chain started at $u$;
$\Phi_{uz}$ is the probability that the first vertex of $T$ visited by the killed chain started at
$u$ is $z$; $\Psi_{zv}$ is the expected number of visits to $v$ of the killed chain started at $z\in
T$ before it returns to $T$; and $\nott N_{uv}$ is the expected number of visits to $v$ before
entering $T$. Two consequences of Lemma~\ref{lemma:first-passage} for $T=\{u\}$ will be useful in
this case: if $f_u$ is the probability that the killed chain started at $u$ returns to $u$, and
$h_{vu}$ is the probability that the killed chain started at $v$ visits $u$, then
\begin{equation}
	\label{eq:return}
	N_{uu}=\frac1{1-f_u},\qquad N_{vu}=h_{vu}N_{uu}.
\end{equation}

We will often apply the lemma with $T=\{x,y\}$, the endpoints of the new edge, so we will deal only
with $\Phi_{ux}$, $\Phi_{uy}$, $\Psi_{xv}$ and $\Psi_{yv}$. In the following, $A$ denotes the
adjacency matrix of $G$, $A'$ that of $G'$, $\rho(\cdot)$ denotes the spectral radius, and $\mathbf
1$ the vector with all entries equal to one. We shall denote the attenuation factor of Katz's index
and the damping factor of PageRank with $\beta$ and $\alpha$, respectively.

\section{Damped Spectral Centralities}
\label{sec:damped}

Following~\cite{vigna}, a \emph{damped spectral centrality} (called a \emph{damped spectral ranking} in ~\cite{vigna}) is a row vector of the form
\[
	\bm r=\bm v\sum_{k\geq0}\beta^kM^k=\bm v(1-\beta M)^{-1},
\]
where $M$ is a nonnegative matrix indexed by the vertices of the graph, $\bm v$ is a row vector (the
\emph{preference vector}), and $0<\beta<1/\rho(M)$ is a \emph{damping factor}. In the notation of
Section~\ref{sec:walks}, setting $W=\beta M$ we have $\bm r=\bm vN$: that is, $r(z)=\sum_u v_uN_{uz}$ is a sum
over all walks of $M$ ending at $z$, each walk being weighted by $\beta^k$ times the product of the
entries of $M$ along the walk, times the preference of its starting vertex. Katz's index is the
damped spectral centrality when $M=A$ and $\bm v=\mathbf 1$; PageRank with damping factor $\alpha$ is the damped spectral centrality when $\beta=\alpha$, $M$ is the
transition matrix of the uniform random walk and $\bm v$ is a distribution; and, as we
will show (Lemma~\ref{lemma:eig-identity}), the dominant eigenvector of $G+x\adj y$ is the damped
spectral centrality of $G$ with a preference vector concentrated on $x$ and $y$. Preference vectors are
usually nonnegative, but we allow signed ones, because the \emph{increment} of a damped spectral
centrality caused by an update of $M$ turns out to be a damped spectral centrality with a signed
preference vector.

In this section we prove a \emph{comparison theorem} for two damped spectral centralities with the same
matrix, and we derive from it a bound on the relative increments caused by an update of $M$ that
modifies the rows $x$ and $y$, in the spirit of~\cite{BLVRMCM}, where strict rank monotonicity of
damped spectral centralities on directed graphs is proved for updates modifying a single row. All
positive results of the paper about Katz's index, eigenvector centrality and PageRank will be
obtained by instantiating these two results. Throughout the section, $M$ is a nonnegative matrix,
not necessarily symmetric, $0<\beta<1/\rho(M)$, $N=(1-\beta M)^{-1}$, and $x\neq y$ are two fixed
vertices; all walk-related quantities ($\nott N$, $\Phi$, $\Psi$) refer to $W=\beta M$ and
$T=\{x,y\}$, whereas walks are walks of $M$. The \emph{successors} of a vertex $u$ are the vertices $w$
such that $M_{uw}>0$.

The comparison theorem rests on a decomposition that is the whole point of the section. Consider two
damped spectral centralities $\bm q=\bm wN$ and $\bm r=\bm vN$, where $\bm w$ is nonpositive outside of
$\{x,y\}$ and $\bm v$ is nonnegative. Then, for every vertex $z\neq x,y$, $q(z)$ and $r(z)$ are the
\emph{same} nonnegative combination of $q(x)$, $q(y)$ and of $r(x)$, $r(y)$, respectively, up to a
correction accounting for the walks that never pass through $x$ or $y$; and the correction is
nonpositive for $\bm q$ and nonnegative for $\bm r$. Formally:

\begin{lemma}[Decomposition]
	\label{lemma:damped-decomp}
	Let $\bm q=\bm wN$ and $\bm r=\bm vN$, where $w_u\leq0$ for every $u\neq x,y$ and $\bm v\geq0$.
	Then, for every $z\neq x,y$,
	\[
		q(z)=\Psi_{xz}q(x)+\Psi_{yz}q(y)-\gamma_z
		\qquad\text{and}\qquad
		r(z)=\Psi_{xz}r(x)+\Psi_{yz}r(y)+\nu_z,
	\]
	where $\gamma_z=-\sum_u w_u\nott N_{uz}\geq0$ and $\nu_z=\sum_u v_u\nott N_{uz}\geq0$.
	Moreover, if there is no walk from $x$ or $y$ to $z$, then $\Psi_{xz}=\Psi_{yz}=0$, and
	$\gamma_z=0$ provided that $\bm w$ is supported on $x$, $y$ and their successors.
\end{lemma}
\begin{proof}
	By the last-passage decomposition (Lemma~\ref{lemma:first-passage}),
	\begin{multline*}
		q(z)=\sum_u w_uN_{uz}=\sum_u w_u\nott N_{uz}+\Psi_{xz}\sum_u w_uN_{ux}+\Psi_{yz}\sum_u w_uN_{uy}=\\
		=-\gamma_z+\Psi_{xz}q(x)+\Psi_{yz}q(y),
	\end{multline*}
	and $\gamma_z\geq0$ because $\nott N_{uz}=0$ for $u\in\{x,y\}$ and $w_u\leq0$ otherwise. The
	same computation gives the identity for $\bm r$. Finally, a walk from $x$ or $y$ to $z$ has a
	vertex in $\{x,y\}$ followed by a last-passage walk, so if there are no such walks then
	$\Psi_{xz}=\Psi_{yz}=0$; moreover, in that case $\nott N_{uz}=0$ for every successor $u$ of $x$
	or $y$, as a walk from $u$ to $z$ would extend to a walk from $x$ or $y$ to $z$, so
	$\gamma_z=0$ if $\bm w$ is supported on $x$, $y$ and their successors.
\end{proof}
Note that $\gamma_z>0$ iff some vertex $u\neq x,y$ with $w_u<0$ has a walk to $z$ that does not
touch $x$ or $y$, and $\nu_z>0$ iff some vertex $u\neq x,y$ with $v_u>0$ does.

\begin{theorem}[Comparison]
	\label{thm:comparison}
	Let $\bm q$, $\bm r$, $\gamma$ and $\nu$ be as in Lemma~\ref{lemma:damped-decomp}, with
	$r(x),r(y)>0$, and let $z\neq x,y$.
	\begin{renumerate}
		\item\label{enu:comparison-pos} If $r(z)>0$, then
		      \[
			      \frac{q(z)}{r(z)}\leq\max\left(0,\frac{q(x)}{r(x)},\frac{q(y)}{r(y)}\right).
		      \]
		      Moreover, if $\max\bigl(q(x)/r(x),q(y)/r(y)\bigr)>0$, then
		      \[
			      \frac{q(z)}{r(z)}<\max\left(\frac{q(x)}{r(x)},\frac{q(y)}{r(y)}\right)
		      \]
		      provided that $\gamma_z>0$, or $\nu_z>0$, or $\Psi_{xz}=\Psi_{yz}=0$.
		\item\label{enu:comparison-zero} If $r(z)=0$, and $\bm w$ is supported on $x$, $y$ and their successors, then $q(z)=0$.
	\end{renumerate}
\end{theorem}
\begin{proof}
	Let $\mu=\max\bigl(q(x)/r(x),q(y)/r(y)\bigr)$, $a=\Psi_{xz}q(x)+\Psi_{yz}q(y)$ and
	$b=\Psi_{xz}r(x)+\Psi_{yz}r(y)$, so that $q(z)=a-\gamma_z$ and $r(z)=b+\nu_z$ by
	Lemma~\ref{lemma:damped-decomp}.

	\ref{enu:comparison-pos} Assume $r(z)>0$. If $a\leq0$ (in particular, if $\Psi_{xz}=\Psi_{yz}=0$), then $q(z)\leq0$, so
	$q(z)/r(z)\leq0$, hence $q(z)/r(z)<\mu$ if $\mu>0$. If $a>0$, then $b>0$ because $r(x),r(y)>0$,
	and by Lemma~\ref{lemma:mediant}
	\[
		\frac{q(z)}{r(z)}=\frac{a-\gamma_z}{b+\nu_z}\leq\frac ab\leq\mu,
	\]
	where the first inequality is strict if $\gamma_z>0$ or $\nu_z>0$. This proves~(i).

	\ref{enu:comparison-zero} Assume $r(z)=0$. Then $b=0$, so $\Psi_{xz}=\Psi_{yz}=0$, that is, there is no walk from $x$ or $y$
	to $z$; thus $\gamma_z=0$ by Lemma~\ref{lemma:damped-decomp}, and $q(z)=0$.
\end{proof}

We now turn to updates. Let $M'$ be a nonnegative matrix with $\beta<1/\rho(M')$, $N'=(1-\beta
M')^{-1}$, and let $\bm r=\bm vN$, $\bm r'=\bm vN'$ be the damped spectral centralities of $M$ and $M'$
with the same nonnegative preference vector $\bm v$; we let $\bm\Delta=\bm r'-\bm r$. We say that
$M'$ is an \emph{admissible update of $M$ at $x$, $y$} if
\begin{equation}
	\label{eq:update}
	M'=M+\bm e_x^T\bm\delta_x+\bm e_y^T\bm\delta_y,
\end{equation}
where $\bm e_u$ denotes the characteristic row vector of $u$ and $\bm\delta_x$, $\bm\delta_y$ are
row vectors (in other words, only rows $x$ and $y$ change), and
\begin{equation}
	\label{eq:sign}
	(\bm\delta_x)_w\leq0\quad\text{and}\quad(\bm\delta_y)_w\leq0\qquad\text{for every $w\neq x,y$},
\end{equation}
that is, the update can increase only the entries of $M$ lying in the columns $x$ and $y$. Note that
if $(\bm\delta_x)_w<0$ then $M_{xw}>0$, since $M'$ is nonnegative; thus $\bm\delta_x$ and $\bm\delta_y$ can be nonzero only for $x$, $y$ and their successors. We shall see that adding an edge $x\adj y$ to an
undirected graph is an admissible update at $x$, $y$ both for Katz's index and for PageRank. The
\emph{update vector} of an admissible update is
\[
	\bm g=r'(x)\bm\delta_x+r'(y)\bm\delta_y,
\]
which is nonpositive outside of $\{x,y\}$ and may be nonzero only on $x$, $y$ and their successors.

\begin{lemma}[Increment]
	\label{lemma:damped-increment}
	If $M'$ is an admissible update of $M$ at $x$, $y$, then $\bm\Delta=\beta\bm gN$: the increment
	is the damped spectral centrality of $M$ with preference vector $\beta\bm g$.
\end{lemma}
\begin{proof}
	From $\bm r'(1-\beta M')=\bm v=\bm r(1-\beta M)$, adding and subtracting $\beta \bm r' M$ from the left-hand side we obtain $\bm r'(1-\beta M)-\beta \bm r'(M'-M)=\bm r(1-\beta M)$, or equivalently $\bm r'(1-\beta M)=\bm r(1-\beta M)+ \beta \bm r'(M'-M)$. Multiplying both sides by $N=(1-\beta M)^{-1}$ we obtain $\bm r'=\bm r+\beta\bm r'(M'-M)N$,
	and $\bm r'(M'-M)=r'(x)\bm\delta_x+r'(y)\bm\delta_y=\bm g$ by~(\ref{eq:update}). So $\bm r'=\bm r + \beta \bm g N$, and the lemma follows.
\end{proof}

Thus, Lemma~\ref{lemma:damped-decomp} and Theorem~\ref{thm:comparison} apply to $\bm q=\bm\Delta$
and $\bm r$; in the following, $\gamma$ and $\nu$ denote the corrections of
Lemma~\ref{lemma:damped-decomp} for $\bm w=\beta\bm g$ and $\bm v$. Since
$r'(z)/r(z)=1+\Delta_z/r(z)$, we obtain:

\begin{theorem}[Updates]
	\label{thm:damped}
	Let $M'$ be an admissible update of $M$ at $x$, $y$, assume $r(x),r(y)>0$, and let $z\neq x,y$.
	\begin{renumerate}
		\item If $r(z)>0$, then
		      \[
			      \frac{r'(z)}{r(z)}\leq\max\left(1,\frac{r'(x)}{r(x)},\frac{r'(y)}{r(y)}\right).
		      \]
		      Moreover, if $\max\bigl(r'(x)/r(x),r'(y)/r(y)\bigr)>1$, then
		      \[
			      \frac{r'(z)}{r(z)}<\max\left(\frac{r'(x)}{r(x)},\frac{r'(y)}{r(y)}\right)
		      \]
		      provided that $\gamma_z>0$, or $\nu_z>0$, or $\Psi_{xz}=\Psi_{yz}=0$.
		\item If $r(z)=0$, then $r'(z)=0$.
	\end{renumerate}
\end{theorem}
\begin{proof}
	Using Lemma~\ref{lemma:damped-increment}, we have $\bm\Delta=\beta\bm gN$, so we can
	apply Theorem~\ref{thm:comparison} with $\bm q=\bm\Delta$, noting that
	\[
		\max\left(\frac{r'(x)}{r(x)},\frac{r'(y)}{r(y)}\right)=1+\max\left(\frac{\Delta_x}{r(x)},\frac{\Delta_y}{r(y)}\right). \qed
	\]
\end{proof}

In words: no vertex can gain, in relative terms, more than the best of the two endpoints, and if
neither endpoint gains score then no vertex does. The rank-theoretic consequence, by the argument of
Theorem~\ref{thm:ratio-semi-mon}, is:

\begin{corollary}
	\label{cor:damped-rank}
	Let $M'$ be an admissible update of $M$ at $x$, $y$, assume $r(x),r(y)>0$, and that $r'(x)>r(x)$
	or $r'(y)>r(y)$; let $e\in\{x,y\}$ maximize $r'(e)/r(e)$. Then, for every $z\neq x,y$ with
	$r(z)\leq r(e)$, we have $r'(z)\leq r'(e)$, and $r'(z)<r'(e)$ if $r(z)<r(e)$ or if the condition
	for strict inequality of Theorem~\ref{thm:damped} holds for $z$. In particular, $\bm r$ is rank
	semi-monotone at $x$, $y$ (Definition~\ref{def:rank-semi}), and strictly so if the condition
	holds for every $z\neq x,y$ with $r(z)>0$.
\end{corollary}
\begin{proof}
	If $r(z)=0$ then $r'(z)=0<r'(e)$ by Theorem~\ref{thm:damped}(ii). Otherwise
	$r'(z)=r(z)\,\bigl(r'(z)/r(z)\bigr)\leq r(z)\,\bigl(r'(e)/r(e)\bigr)\leq r'(e)$, the first
	inequality being strict under the condition and the second one if $r(z)<r(e)$.
\end{proof}

As in~\cite{BLVRMCM}, the hypothesis that at least one endpoint gains score is trivial for Katz's
index, and for PageRank it follows from the following corollary, the analogue of the observation of
Section~\ref{sec:defs} that strict ratio semi-monotonicity implies score semi-monotonicity when the
sum of the scores does not decrease:

\begin{corollary}
	\label{cor:damped-score}
	Let $M'$ be an admissible update of $M$ at $x$, $y$ such that $\sum_zr'(z)\geq\sum_zr(z)$ (e.g.,
	because $M$ and $M'$ are row-stochastic and $\bm v$ is a distribution). If $\gamma_u>0$ for
	some $u\neq x,y$, then $r'(x)>r(x)$ or $r'(y)>r(y)$.
\end{corollary}
\begin{proof}
	If $\Delta_x,\Delta_y\leq0$, Lemma~\ref{lemma:damped-decomp} gives
	$\Delta_z\leq-\gamma_z\leq0$ for every $z\neq x,y$, with $\Delta_u<0$, so
	$\sum_z\Delta_z<0$.
\end{proof}

We remark that nothing in this section requires $M$ to be symmetric, so Theorem~\ref{thm:damped}
applies also to the addition of an arc $x\to y$ to a directed graph, in which case $\bm\delta_y=0$.
If also $(\bm\delta_x)_x\leq0$, as for Katz's index and PageRank, and $\Delta_y\geq0$ (which is always
true for these two measures~\cite{BLVRMCM}), then the largest relative increase is attained at the
target $y$. Indeed, the first-passage decomposition with $T=\{y\}$ gives
$N_{uy}N_{yx}=\Phi_{uy}N_{yy}N_{yx}\leq N_{yy}N_{ux}$ for every $u$, whence $r(y)N_{yx}\leq
N_{yy}r(x)$ and, as $\bm g$ is nonpositive outside of $y$,
$N_{yy}\Delta_x=\beta\sum_w g_wN_{yy}N_{wx}\leq\beta\sum_w g_wN_{wy}N_{yx}=N_{yx}\Delta_y$; if
$N_{yx}=0$ then $\Delta_x\leq0\leq\Delta_y$, and otherwise $\Delta_x/r(x)\leq
N_{yx}\Delta_y/(N_{yy}r(x))\leq\Delta_y/r(y)$. Thus, if $r(x),r(y)>0$, Theorem~\ref{thm:damped}
yields $r'(z)/r(z)\leq r'(y)/r(y)$ for every $z$ with $r(z)>0$, a ``ratio monotonicity'' property of
the directed case which implies, by the argument of Theorem~\ref{thm:ratio-semi-mon}, the rank
monotonicity results of~\cite{BLVRMCM} (in strict form, under the condition for strict inequality of
Theorem~\ref{thm:damped}).

\section{Katz's Index}
\label{sec:katz}

Katz's index~\cite{katz} with attenuation factor $\beta$ and preference vector $\bm v$ is
\[
	\bm c=\bm v\sum_{k\geq0}\beta^kA^k=\bm v(1-\beta A)^{-1},
\]
where $0<\beta<1/\rho(A)$; the original definition has $\bm v=\mathbf 1$, whereas
Hubbell~\cite{hubbell} allows for arbitrary positive preference vectors. It is the damped spectral
centrality of Section~\ref{sec:damped} for $M=A$, and $c(z)=\sum_u v_uN_{uz}$, where $N=(1-\beta
A)^{-1}$ is the matrix of Section~\ref{sec:walks} for $W=\beta A$; in other words, $c(z)$ is the sum
over all walks ending at $z$ of $\beta$ raised to the length of the walk, weighted by the preference
of the starting vertex. Since $\rho(A)\leq\rho(A')$, in order for both $\bm c$ and $\bm c'$ to be
defined we assume $0<\beta<1/\rho(A')$. As $A'=A+\bm e_x^T\bm e_y+\bm e_y^T\bm e_x$, adding the edge
$x\adj y$ is an admissible update at $x$, $y$ with $\bm\delta_x=\bm e_y$ and $\bm\delta_y=\bm e_x$,
whose update vector is $\bm g=c'(x)\bm e_y+c'(y)\bm e_x$. Lemma~\ref{lemma:damped-increment} then
reads (recall that $N$ is symmetric):

\begin{lemma}
	\label{lemma:katz-increment}
	For every vertex $z$,
	\[
		c'(z)-c(z)=\beta\bigl(N_{zx}c'(y)+N_{zy}c'(x)\bigr).
	\]
\end{lemma}
\begin{proof}
	By Lemma~\ref{lemma:damped-increment},
	$c'(z)-c(z)=\beta\sum_w g_wN_{wz}=\beta\bigl(c'(x)N_{yz}+c'(y)N_{xz}\bigr)$, and $N$ is
	symmetric.
\end{proof}
The lemma has a simple combinatorial reading: a walk in $G'$ that is not a walk in $G$ can be split
at its first traversal of the new edge into a walk in $G$ ending in $x$ (or $y$), followed by the
new edge and by an arbitrary walk in $G'$ starting from $y$ (or $x$).

\begin{theorem}
	\label{thm:katz}
	Let $G$ be an undirected graph, $0<\beta<1/\rho(A')$, and $\bm v$ be a nonnegative preference
	vector such that $c(x),c(y)>0$. Then:
	\begin{renumerate}
		\item $c'(z)\geq c(z)$ for every vertex $z$, and $c'(x)>c(x)$, $c'(y)>c(y)$;
		\item for every $z\neq x,y$ with $c(z)>0$,
		      \[
			      \frac{c'(z)}{c(z)}\leq\max\left(\frac{c'(x)}{c(x)},\frac{c'(y)}{c(y)}\right),
		      \]
		      with strict inequality if $v_z>0$; 
			  if $c(z)=0$ then $c'(z)=0$.
	\end{renumerate}
	In particular, if $\bm v$ is everywhere positive (e.g., $\bm v=\mathbf 1$), Katz's index is
	score monotone and strictly ratio semi-monotone on all undirected graphs, for every
	$0<\beta<1/\rho(A')$, and thus it is score semi-monotone and strictly rank semi-monotone on the
	same graphs. If $\bm v$ is just nonnegative and all scores are nonzero, Katz's index is ratio
	semi-monotone and rank semi-monotone.
\end{theorem}
\begin{proof}
	(i). Lemma~\ref{lemma:katz-increment} gives $\bm c'\geq\bm c$, and $c'(x)-c(x)\geq\beta c'(y)N_{xx}$; since $N_{xx}\geq1$ (because of the empty walk) and $c'(z)\geq c(z)$ for all $z$, we deduce $c'(x)-c(x)\geq\beta c(y)>0$; symmetrically for $y$. 
	
	\noindent (ii). We apply Theorem~\ref{thm:damped}, noting that $\nu_z=\sum_u v_u \nott N_{uz}\geq v_z\nott N_{zz}\geq v_z$. 
	
	\noindent The last
	statements follow from~(i), (ii) and Theorem~\ref{thm:ratio-semi-mon}, as $\bm c\geq\bm v$.
\end{proof}

Strictness may fail if the preference vector has zero entries: if $z$ is a vertex adjacent only
to $x$ and $\bm v=\bm e_x$, then $c(z)=\beta c(x)$ and $c'(z)=\beta c'(x)$, because every walk from
$x$ to $z$ is a walk from $x$ to $x$ followed by the edge $x\adj z$, so $c'(z)/c(z)=c'(x)/c(x)$.

We remark that Katz's index is \emph{not} $\delta$-semi-monotone, which shows that the additive and
the multiplicative variants are genuinely different properties. Consider the star $G=K_{1,n}$, for
$n\geq3$, with center $z$, and let $x$ and $y$ be two leaves. Every closed walk of $G$ from $z$ is a
sequence of round trips to a leaf, so $N_{zz}=\sum_k(n\beta^2)^k=1/(1-n\beta^2)$; moreover,
$N_{zx}=N_{zy}=\beta N_{zz}$, $N_{xx}=1+\beta^2N_{zz}$ and $N_{xy}=\beta^2N_{zz}$, and $c'(x)=c'(y)$ by
symmetry, so by Lemma~\ref{lemma:katz-increment}
\[
	\Delta_z=c'(z)-c(z)=\beta \left(N_{zx}c'(y)+N_{zy}c'(x)\right)=2\beta^2N_{zz}c'(x), 
\]
and
\[
	\Delta_x=c'(x)-c(x)=\beta \left(N_{xx}c'(y)+N_{xy}c'(x)\right)=\beta c'(x)\left(1+2\beta^2N_{zz}\right).
\]
Subtracting the two equations we obtain
\begin{multline*}
	\Delta_z-\Delta_x=\beta c'(x)\bigl(2\beta N_{zz}-1-2\beta^2N_{zz}\bigr)=\\
	=\beta c'(x)\left(\frac{2\beta}{1-n\beta^2}-1-\frac{2\beta^2}{1-n\beta^2}\right)
	=\beta c'(x)\,\frac{(n-2)\beta^2+2\beta-1}{1-n\beta^2},
\end{multline*}
which is positive iff $\beta>1/\bigl(1+\sqrt{n-1}\bigr)$ (the denominator is positive because $\rho(A')>\rho(A)=\rho(K_{1,n})=\sqrt{n}$ so $\beta<1/\rho(A')<1/\sqrt{n}$). 

\smallskip
To bound $\rho(A')$, note that, by symmetry, the dominant eigenvector of $A'$ (which is unique) has the same value, say $a$, on all the leaves of the star except for $x$ and $y$, some other value, say $b$, on $x$ and $y$, and some value, say $c$, at the center. The eigenvector equations for the eigenvalue $\rho$ give:
\begin{align*}
	\rho a &= c\\
	\rho b &= b + c\\
	\rho c &= 2b + (n-2) a.
\end{align*}
Solving this system we have $a=c/\rho$, $b=c/(\rho-1)$, and substituting into the third equation gives $\rho c = 2c/(\rho-1) + (n-2)c/\rho$, which simplifies to the polynomial equation $p(\rho)=\rho^3 - \rho^2 - n\rho + n - 2 = 0$: thus, $\rho(A')$ is a root of $p$. Writing $s=\sqrt{n-1}$, a direct computation gives $p(1+s)=2(n-2)>0$ and $p'(\rho)=3\rho^2-2\rho-n>0$ for $\rho\geq1+s$, so $p$ is positive and increasing on $[1+\sqrt{n-1}\..\infty)$, which implies that $\rho(A')<1+\sqrt{n-1}$.

\smallskip
Hence, for every
$\beta$ in the nonempty interval $\bigl(1/(1+\sqrt{n-1})\..1/\rho(A')\bigr)$ the center gains more
than both endpoints of the new edge. For instance, for $n=3$ this happens for
$\beta\in(0.4142\..0.4608)$. The same phenomenon can be observed numerically for eigenvector
centrality and PageRank on larger graphs.

\section{Eigenvector Centrality}
\label{sec:eigenvector}

Let $G$ be connected, and let $\lambda=\rho(A)$ and $\lambda'=\rho(A')$: since $A'\geq A$ and
$A'\neq A$, by the Perron--Frobenius theorem~\cite{berman} we must have $\lambda'>\lambda$. Eigenvector
centrality~\cite{landau,berman} is the (unique, positive) dominant eigenvector $\bm v$ of $A$,
defined up to a scaling factor. Rank monotonicity properties do not depend on the scaling factor,
whereas score monotonicity properties depend on the normalization chosen; we will consider, as
in~\cite{boldi_furia_vigna_2023}, normalization in norm $\ell_1$, in norm $\ell_2$, and the
projection of the vector $\mathbf 1$ on the dominant eigenspace (i.e., $\bm v\,\langle\mathbf
1,\bm v\rangle$ for $\|\bm v\|_2=1$).

\begin{lemma}
	\label{lemma:eig-identity}
	Let $R=(\lambda'-A)^{-1}=\sum_{k\geq0}A^k/(\lambda')^{k+1}$. Then, for every vertex $z$,
	\[
		v'(z)=R_{zx}v'(y)+R_{zy}v'(x)\qquad\text{and}\qquad v(z)=(\lambda'-\lambda)\sum_uR_{zu}v(u).
	\]
\end{lemma}
\begin{proof}
	Since $A'\bm v'=\lambda'\bm v'$, we have $(\lambda'-A)\bm v'=\lambda' \bm v'-A\bm v'=(A'-A)\bm v'$, and the latter
	vector has entry $v'(y)$ at $x$, entry $v'(x)$ at $y$, and is zero elsewhere. Multiplying by $R$
	gives the first identity. The second identity follows from
	$(\lambda'-A)\bm v=(\lambda'-\lambda)\bm v$.
\end{proof}
Note that the series defining $R$ converges because $\lambda'>\lambda=\rho(A)$, and that
$R=(\lambda')^{-1}(1-A/\lambda')^{-1}$: thus, up to the factor $1/\lambda'$, $R$ is the matrix $N$
of Section~\ref{sec:walks} for $W=A/\lambda'$. Since $R$ is symmetric, the lemma says that (the transposes of) $\bm v'$ and $\bm v$ are the damped spectral centralities of the \emph{old} graph with
damping factor $1/\lambda'$ and preference vectors $v'(y)\bm e_x+v'(x)\bm e_y$ and
$(\lambda'-\lambda)\bm v$, respectively: the new eigenvector is a Katz index of the old graph with a
preference vector concentrated on $x$ and $y$, and the old eigenvector is a Katz index of the old
graph with a positive preference vector. We are therefore in the setting of the comparison theorem
of Section~\ref{sec:damped}.

\begin{theorem}
	\label{thm:eigenvector}
	Eigenvector centrality is strictly ratio semi-monotone on connected undirected graphs,
	independently of the normalization chosen. Thus, it is strictly rank semi-monotone on the same
	graphs.
\end{theorem}
\begin{proof}
	The ratio $v'(z)/v(z)$ changes by the same multiplicative constant for all $z$ when 
	normalization is applied, so we can fix arbitrary normalizations. We apply
	Theorem~\ref{thm:comparison} with $M=A$, $\beta=1/\lambda'$, $\bm q=\bm v'$ and $\bm r=\bm v$
	(scaling the preference vectors is irrelevant). The hypotheses are satisfied because $\bm v$ is
	positive and $v'(y)\bm e_x+v'(x)\bm e_y$ is nonnegative and can be non-zero only on $\{x,y\}$; moreover
	$\max\bigl(v'(x)/v(x),v'(y)/v(y)\bigr)>0$ and, for every $z\neq x,y$,
	$\nu_z\geq(\lambda'-\lambda)v(z)\nott R_{zz}>0$, as $\nott R_{zz}\geq1/\lambda'$. Hence
	$v'(z)/v(z)<\max\bigl(v'(x)/v(x),v'(y)/v(y)\bigr)$ for every $z\neq x,y$, which is exactly
	strict ratio semi-monotonicity.
\end{proof}

A normalization can be seen as a function $\|\cdot\|$ from positive vectors to positive reals,
such that $\|t\cdot \bm v\|=t\cdot \|\bm v\|$ for all vectors $\bm v$ and all $t>0$. The normalization of a vector $\bm v$ is $\bm v/\|\bm v\|$. We say that the normalization is
\emph{monotone} if $\bm u\leq\bm w$ and $\bm u\neq\bm w$ imply $\|\bm u\|<\|\bm w\|$. Norms $\ell_1$
and $\ell_2$ (indeed, all $\ell_p$ norms with $p<\infty$) are monotone.

\begin{corollary}
	\label{cor:eig-score}
	Eigenvector centrality is score semi-monotone on connected undirected graphs for every monotone
	normalization; in particular, for norms $\ell_1$ and $\ell_2$.
\end{corollary}
\begin{proof}
	Assume by contradiction that $v'(x)\leq v(x)$ and $v'(y)\leq v(y)$, where $\bm v$ and $\bm v'$
	are normalized. By Theorem~\ref{thm:eigenvector}, $v'(z)/v(z)<1$ for every $z\neq x,y$, so
	$\bm v'\leq\bm v$ and $\bm v'\neq\bm v$ (there is at least one vertex $z\neq x,y$ as $G$ is
	connected and $x\noadj y$). Thus $\|\bm v'\|<\|\bm v\|$, whereas both norms should be one.
\end{proof}

The projection of $\mathbf 1$ on the dominant eigenspace is not a monotone normalization, and indeed
in that case score semi-monotonicity fails:
\begin{theorem}
	\label{thm:eig-proj}
	Eigenvector centrality normalized by projecting $\mathbf 1$ on the dominant eigenspace is not
	score semi-monotone on connected undirected graphs.
\end{theorem}
\begin{proof}
	Consider the graph of Figure~\ref{fig:eigproj}. Before adding the edge $x\adj y$, the score of
	$x$ and $y$ is $1.16684\ldots$; after adding it, it is $1.09054\ldots$ (the computation can be
	carried out exactly, as $\lambda$ and $\lambda'$ are algebraic numbers of degree $9$ and $8$,
	respectively).
\end{proof}

\begin{figure}
	\centering
	\begin{tikzpicture}[main/.style = {draw, circle, minimum size=5mm, inner sep=0pt}, scale=0.9]
		\node[main] (x) at (0,1) {$x$};
		\node[main] (y) at (0,-1) {$y$};
		\node[main] (h) at (1.2,0) {};
		\node[main] (k) at (-1.2,0) {};
		\node[main] (a1) at (2.4,0) {};
		\node[main] (a2) at (3.6,0) {};
		\node[main] (a3) at (4.8,0) {};
		\node[main] (a4) at (6,0) {};
		\node[main] (c) at (7.2,0) {};
		\foreach \i in {1,...,6} {
				\node[main] (l\i) at ({7.2+1.3*cos(75-30*(\i-1))},{1.3*sin(75-30*(\i-1))}) {};
				\draw (c) -- (l\i);
			}
		\draw (x) -- (h) -- (y) -- (k) -- (x);
		\draw (h) -- (k);
		\draw (h) -- (a1) -- (a2) -- (a3) -- (a4) -- (c);
		\draw[dashed] (x) -- (y);
	\end{tikzpicture}
	\caption{\label{fig:eigproj}A counterexample to score semi-monotonicity for eigenvector
	centrality normalized by projecting $\mathbf 1$ on the dominant eigenspace: a $4$-clique minus
	the edge $x\adj y$, with a path of length five attached to one of its vertices and ending in the
	center of a star with six leaves. When the dashed edge $x\adj y$ is added, the scores of both
	$x$ and $y$ decrease from $1.16684$ to $1.09054$. With normalization in norm $\ell_1$ the scores
	increase from $0.11875$ to $0.20869$, and with normalization in norm $\ell_2$ from $0.37224$ to
	$0.47706$, as predicted by Corollary~\ref{cor:eig-score}.}
\end{figure}
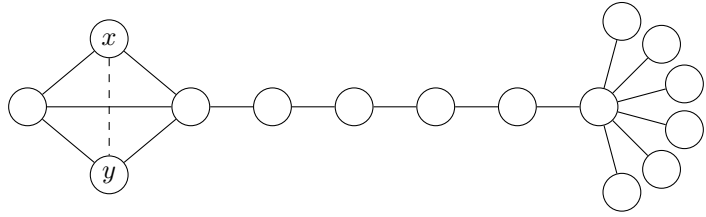

The phenomenon behind Theorem~\ref{thm:eig-proj} is easy to explain. If $\hat{\bm v}$ denotes the
$\ell_2$-normalized dominant eigenvector, the projection normalization gives scores $\hat
v(z)\langle\mathbf1,\hat{\bm v}\rangle$. In the graph of Figure~\ref{fig:eigproj} the star is almost
as ``strong'' as the $4$-clique minus an edge (their spectral radii are $\sqrt6\approx2.449$ and
$(1+\sqrt{17})/2\approx2.562$, respectively), so before the addition of the new edge the dominant
eigenvector has a long, slowly decreasing tail on the path and on the star, which carries a large
part of its $\ell_1$ norm: $\langle\mathbf1,\hat{\bm v}\rangle\approx3.13$. Once the edge is added,
the clique becomes much stronger (its spectral radius is now $3$), the tail collapses and
$\langle\mathbf1,\hat{\bm v}'\rangle\approx2.29$. The decrease of the factor
$\langle\mathbf1,\hat{\bm v}\rangle$ overwhelms the increase of $\hat v(x)$ and $\hat v(y)$
guaranteed by Corollary~\ref{cor:eig-score}. Larger and more robust examples can be obtained by
attaching, in place of the star, any large graph whose spectral radius is slightly smaller than that
of the clique minus an edge: in fact, by doing this the decrease of the scores of $x$ and $y$ can be made
arbitrarily large.

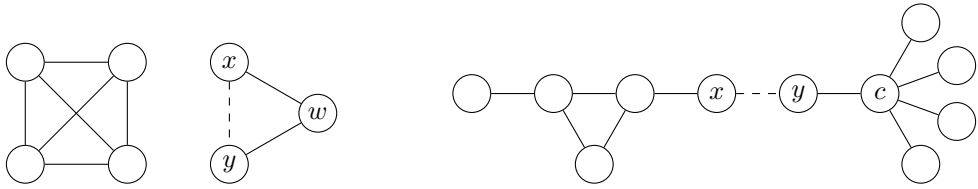
\begin{figure}
	\centering
	\begin{tikzpicture}[main/.style = {draw, circle, minimum size=5mm, inner sep=0pt}, scale=0.9]
		% K_4
		\node[main] (k1) at (0,0.75) {};
		\node[main] (k2) at (1.5,0.75) {};
		\node[main] (k3) at (0,-0.75) {};
		\node[main] (k4) at (1.5,-0.75) {};
		\draw (k1) -- (k2) -- (k4) -- (k3) -- (k1);
		\draw (k1) -- (k4);
		\draw (k2) -- (k3);
		% path x - w - y
		\node[main] (x) at (3,0.75) {$x$};
		\node[main] (w) at (4.3,0) {$w$};
		\node[main] (y) at (3,-0.75) {$y$};
		\draw (x) -- (w) -- (y);
		\draw[dashed] (x) -- (y);
	\end{tikzpicture}
	\qquad\qquad
	\begin{tikzpicture}[main/.style = {draw, circle, minimum size=5mm, inner sep=0pt}, scale=0.9]
		% bull: equilateral triangle b1 b2 b3 (side 1.2), pendants p (at b1) and x (at b3)
		\node[main] (p) at (0,0) {};
		\node[main] (b1) at (1.2,0) {};
		\node[main] (b3) at (2.4,0) {};
		\node[main] (b2) at (1.8,-1.0392) {};
		\node[main] (x) at (3.6,0) {$x$};
		\draw (p) -- (b1) -- (b2) -- (b3) -- (b1);
		\draw (b3) -- (x);
		% star K_{1,5} with center c and leaf y
		\node[main] (y) at (4.8,0) {$y$};
		\node[main] (c) at (6,0) {$c$};
		\foreach \i in {1,...,4} {
				\node[main] (l\i) at ({6+1.2*cos(60-40*(\i-1))},{1.2*sin(60-40*(\i-1))}) {};
				\draw (c) -- (l\i);
			}
		\draw (y) -- (c);
		\draw[dashed] (x) -- (y);
	\end{tikzpicture}
	\caption{\label{fig:eigdisc}Two graphs that are not connected on which eigenvector centrality is
	not semi-monotone. Left: $K_4$ and a path $x\adj w\adj y$; adding the dashed edge $x\adj y$ does
	not change the dominant eigenvector, which vanishes on the triangle. Right: a bull and a star
	$K_{1,5}$; before adding the dashed edge $x\adj y$ the dominant eigenvector vanishes on the
	star, but afterwards the center $c$ of the star has a larger score than both $x$ and $y$.}
\end{figure}

We conclude this section by showing that, contrarily to what happens for PageRank
(Section~\ref{sec:pagerank}), connectivity cannot be replaced by a condition on the scores. If $G$
is not connected, eigenvector centrality is well defined only if $G$ has a unique connected
component of maximum spectral radius, and in that case it is the Perron vector of that component,
extended with zeros; in particular, it cannot be everywhere nonzero. Let $G$ be the disjoint union
of $K_4$ and of a path $x\adj w\adj y$ (Figure~\ref{fig:eigdisc}, left): adding $x\adj y$ turns the
path into a triangle, whose spectral radius $2$ is smaller than $\rho(K_4)=3$, so $\bm v'=\bm v$
under every normalization, and both score semi-monotonicity and strict rank semi-monotonicity fail,
exactly as they do for PageRank when $p(x)=p(y)=0$. Rank semi-monotonicity may fail even if
$v(x)>0$: let $G$ be the disjoint union of a \emph{bull} (a triangle with two pendant vertices) and
of a star $K_{1,5}$, let $x$ be a pendant vertex of the bull, and let $y$ be a leaf of the star,
whose center we denote by $c$ (Figure~\ref{fig:eigdisc}, right). Since the spectral radius of the
bull is $(1+\sqrt{13})/2\approx2.303>\sqrt5=\rho(K_{1,5})$, $\bm v$ is null on the star. On the
other hand, $G'$ is connected with $\rho(A')\approx2.399$, and normalizing in $\ell_1$ one finds
$v'(c)\approx0.134$, $v'(x)\approx0.101$ and $v'(y)\approx0.098$ (the inequalities
$v'(c)>v'(x)>v'(y)$ can be verified exactly by isolating the largest root of the characteristic
polynomial of $A'$). Thus $v(c)=v(y)=0<v(x)$, but the center of the star overtakes both endpoints of
the new edge, and rank semi-monotonicity fails at $x$, $y$ (score semi-monotonicity holds, as $y$
gains score).

\section{PageRank}
\label{sec:pagerank}

In this section $G$ is an undirected graph without isolated vertices, but not necessarily connected:
we will see that connectivity can be replaced by the weaker requirement that all scores be nonzero,
which is the hypothesis of the rank-monotonicity results of~\cite{boldi_furia_vigna_2023}. Let $D$
be the diagonal matrix of degrees and $P=D^{-1}A$ be the row-stochastic transition matrix of the
uniform random walk on $G$. PageRank~\cite{pagerank} with damping factor $\alpha\in(0\..1)$ and
preference vector $\bm v$ (a nonnegative row vector with unit $\ell_1$ norm) is the row vector
\[
	\bm p=(1-\alpha)\,\bm v\sum_{k\geq0}\alpha^kP^k=(1-\alpha)\,\bm v(1-\alpha P)^{-1}.
\]
Let $N=(1-\alpha P)^{-1}$ be the matrix of Section~\ref{sec:walks} for $W=\alpha P$, so that
$N_{uz}$ is the expected number of visits to $z$ of the random walk started at $u$ and killed with
probability $1-\alpha$ at each step. All entries of $N$ are nonnegative, and $N_{uz}>0$ iff $z$ is
reachable from $u$; hence $p(z)=(1-\alpha)\sum_u v_uN_{uz}$ is positive iff $\bm v$ has positive
mass on the connected component of $z$. In particular, all scores are nonzero if $G$ is connected,
or if $\bm v$ is everywhere nonzero. Moreover, the rows of $N$ sum to $1/(1-\alpha)$, so the entries
of $\bm p$ sum to one, in $G$ as well as in $G'$. Finally, since $N=(D-\alpha A)^{-1}D$ and
$D-\alpha A$ is symmetric, we have the \emph{reversibility} identity
\begin{equation}
	\label{eq:reversibility}
	\frac{N_{uz}}{d(z)}=\frac{N_{zu}}{d(u)}.
\end{equation}
The rows $x$ and $y$ of $P'$ are $\bigl(d(x)P_{x\cdot}+\bm e_y\bigr)/(d(x)+1)$ and
$\bigl(d(y)P_{y\cdot}+\bm e_x\bigr)/(d(y)+1)$, where $P_{u\cdot}$ denotes the $u$-th row of $P$ and
$\bm e_u$ the characteristic (row) vector of $u$; hence
\[
	P'=P+\bm e_x^T\bm\delta_x+\bm e_y^T\bm\delta_y,\qquad
	\bm\delta_x=\frac{\bm e_y-P_{x\cdot}}{d(x)+1},\qquad
	\bm\delta_y=\frac{\bm e_x-P_{y\cdot}}{d(y)+1},
\]
and PageRank is, up to the factor $1-\alpha$, the damped spectral centrality of
Section~\ref{sec:damped} for $\beta=\alpha$ and $M=P$. The update is admissible at $x$, $y$: since
$P_{xy}=0$ (as $x\noadj y$), $(\bm\delta_x)_w=-P_{xw}/(d(x)+1)\leq0$ for
every $w\neq x,y$, and symmetrically for $\bm\delta_y$. In the following, $P'$ is the transition
matrix of $G'$, $\bm p'$ the PageRank of $G'$, and
\[
	s_x=\frac{p'(x)}{d(x)+1},\qquad s_y=\frac{p'(y)}{d(y)+1}.
\]
Note that $\alpha s_x$ and $\alpha s_y$ are the amounts of PageRank flowing in $G'$ along the new edge
from $x$ to $y$ and from $y$ to $x$, respectively, and that the update vector of
Section~\ref{sec:damped} (computed for $\bm p$ rather than $\bm r=\bm p/(1-\alpha)$, which changes
nothing, as Lemma~\ref{lemma:damped-increment} is homogeneous) is
\[
	\bm g=s_x\bigl(\bm e_y-P_{x\cdot}\bigr)+s_y\bigl(\bm e_x-P_{y\cdot}\bigr).
\]

\begin{lemma}
	\label{lemma:pr-increment}
	Denoting with $N_{u\cdot}$ the $u$-th row of $N$,
	\[
		\bm p'=\bm p+s_x\bigl(\bm e_x+\alpha N_{y\cdot}-N_{x\cdot}\bigr)+s_y\bigl(\bm e_y+\alpha N_{x\cdot}-N_{y\cdot}\bigr).
	\]
	Consequently, letting $\Delta_z=p'(z)-p(z)$,
	\begin{align*}
		\Delta_z & =(\alpha s_x-s_y)N_{yz}+(\alpha s_y-s_x)N_{xz}\qquad\text{for $z\neq x,y$}, \\
		\Delta_x & =s_x+(\alpha s_x-s_y)N_{yx}+(\alpha s_y-s_x)N_{xx},                         \\
		\Delta_y & =s_y+(\alpha s_x-s_y)N_{yy}+(\alpha s_y-s_x)N_{xy}.
	\end{align*}
\end{lemma}
\begin{proof}
	By Lemma~\ref{lemma:damped-increment},
	$\bm p'=\bm p+\alpha\bm gN=\bm p+\alpha\bigl[s_x(\bm e_y-P_{x\cdot})+s_y(\bm e_x-P_{y\cdot})\bigr]N$,
	and $\alpha PN=N-1$ (because $(1-\alpha P)N=1$), so $\alpha(\bm e_y-P_{x\cdot})N=\alpha N_{y\cdot}-N_{x\cdot}+\bm e_x$, and
	symmetrically for the other term. The remaining formulas in the statement follow by direct computation.
\end{proof}

Note that $s_x,s_y\geq0$, and that $s_x>0$ iff $s_y>0$ iff $p(x)>0$ or $p(y)>0$. Indeed, if $s_x=0$ then
$p'(x)=0$, and the formula for $\Delta_x$ gives $p(x)+s_y(\alpha N_{xx}-N_{yx})=0$; but
$N_{yx}=h_{yx}N_{xx}\leq\alpha^2N_{xx}$ by~(\ref{eq:return}), because $x\noadj y$, so $\alpha
N_{xx}-N_{yx}\geq\alpha(1-\alpha)N_{xx}>0$, and we conclude that $p(x)=0$ and $s_y=0$. Symmetrically,
$s_y=0$ implies $p(y)=0$ and $s_x=0$. Conversely, if $p(x)=p(y)=0$ no walk from any node $z$ such that $v_z>0$
reaches $x$ or $y$ in $G$, hence also in $G'$ (whose only new edge is $x\adj y$), so $s_x=s_y=0$, and in
this case Lemma~\ref{lemma:pr-increment} gives $\bm p'=\bm p$.

Score semi-monotonicity now follows as an application of Corollary~\ref{cor:damped-score}:

\begin{theorem}
	\label{thm:pr-score}
	Let $G$ be an undirected graph without isolated vertices, $\alpha\in(0\..1)$, $\bm v$ a
	preference vector, and $x\noadj y$ be such that $p(x)>0$ or $p(y)>0$. Then $p'(x)>p(x)$ or
	$p'(y)>p(y)$. In particular, PageRank is score semi-monotone on every undirected graph without
	isolated vertices, for every damping factor and every preference vector, provided that all
	scores are nonzero (e.g., if $G$ is connected, or if $\bm v$ is everywhere nonzero). The
	condition on $x$ and $y$ cannot be dropped, as $\bm p'=\bm p$ when $p(x)=p(y)=0$.
\end{theorem}
\begin{proof}
	The entries of $\bm p$ and $\bm p'$ sum to one. A neighbor $w$ of $x$ exists, and $w\neq x,y$;
	since $s_x>0$ (remark following Lemma~\ref{lemma:pr-increment}), $\gamma_w\geq-\alpha g_w\nott
	N_{ww}\geq\alpha s_x/d(x)>0$, and Corollary~\ref{cor:damped-score} applies. The last statement is
	again the remark following Lemma~\ref{lemma:pr-increment}.
\end{proof}

The general theory does not say \emph{which} endpoint gains score. For PageRank, the explicit form
of Lemma~\ref{lemma:pr-increment} gives a precise answer. The formula for $\Delta_z$ shows that at
most one of the two coefficients $\alpha s_x-s_y$ and $\alpha s_y-s_x$ can be positive (if both were, i.e., if $\alpha s_x>s_y$ and $\alpha s_y>s_x$, then
multiplying the two inequalities we would get $\alpha^2 s_xs_y>s_xs_y$). This gives the following trichotomy, which refines
Theorem~\ref{thm:pr-score}:

\begin{theorem}
	\label{thm:pr-trichotomy}
	Let $G$ be an undirected graph without isolated vertices, $\alpha\in(0\..1)$, $\bm v$ a
	preference vector, and $x\noadj y$ be such that $p(x)>0$ or $p(y)>0$. Then:
	\begin{renumerate}
		\item if $\alpha s_x\leq s_y\leq s_x/\alpha$, no vertex $z\neq x,y$ gains score, at least
		      one such vertex loses score (every such vertex, if $G$ is connected), and neither $x$
		      nor $y$ loses score: more precisely, $\Delta_x\geq0$, with equality only if $s_y=\alpha
		      s_x$ and all neighbors of $x$ are pendant vertices, and symmetrically for $y$; in
		      particular, the sum of the scores of $x$ and $y$ increases;
		\item if $s_y<\alpha s_x$, then $y$ gains score;
		\item if $s_x/\alpha<s_y$, then $x$ gains score.
	\end{renumerate}
\end{theorem}
\begin{proof}
	Recall that $s_x,s_y>0$ by the remark following Lemma~\ref{lemma:pr-increment}. We first rewrite
	the increments of the endpoints using~(\ref{eq:return}): $N_{yy}=1/(1-f_y)$ and
	$N_{xy}=h_{xy}N_{yy}$, so the last formula of Lemma~\ref{lemma:pr-increment} becomes
	\[
		\Delta_y=N_{yy}\bigl[s_y(1-f_y)+(\alpha s_x-s_y)+(\alpha s_y-s_x)h_{xy}\bigr]=N_{yy}\bigl[s_x(\alpha-h_{xy})+s_y(\alpha h_{xy}-f_y)\bigr],
	\]
	and symmetrically $\Delta_x=N_{xx}\bigl[s_y(\alpha-h_{yx})+s_x(\alpha h_{yx}-f_x)\bigr]$. Now,
	$f_y\leq\alpha^2$ because the walk needs at least two steps to return to $y$ (there are no
	loops), with equality iff the killed walk started at $y$ returns to $y$ with certainty in exactly
	two steps, that is, iff all neighbors of $y$ are pendant vertices; and $h_{xy}\leq\alpha^2<\alpha$
	because $x\noadj y$. The same holds with $x$ and $y$ exchanged.

	In case (i) the two coefficients $\alpha s_x-s_y$ and $\alpha s_y-s_x$ are nonpositive, and they cannot
	be both zero (that would give $\alpha^2=1$), so by Lemma~\ref{lemma:pr-increment} $\Delta_z\leq0$
	for all $z\neq x,y$, with strict inequality whenever $N_{yz}>0$ (if $\alpha s_x<s_y$) or $N_{xz}>0$
	(if $\alpha s_y<s_x$). This is the case for every $z\neq x,y$ if $G$ is connected, and in general for
	some $z\neq x,y$: a neighbor $z$ of $y$ (if $\alpha s_x<s_y$) or of $x$ (if $\alpha s_y<s_x$), which is
	different from $x$ and $y$ because $x\noadj y$ and there are no loops, and satisfies
	$N_{yz}\geq\alpha P_{yz}>0$, respectively $N_{xz}\geq\alpha P_{xz}>0$. Moreover, since
	$s_y\geq\alpha s_x$ and $\alpha-h_{yx}>0$,
	\[
		\Delta_x\geq N_{xx}\bigl[\alpha s_x(\alpha-h_{yx})+s_x(\alpha h_{yx}-f_x)\bigr]=N_{xx}\,s_x\,(\alpha^2-f_x)\geq0,
	\]
	where the first inequality is an equality iff $s_y=\alpha s_x$, and the second one iff
	$f_x=\alpha^2$, that is, iff all neighbors of $x$ are pendant vertices; symmetrically for $y$.
	Finally, since the entries of $\bm p$ and $\bm p'$ sum to one,
	$\Delta_x+\Delta_y=-\sum_{z\neq x,y}\Delta_z>0$.

	In case (ii), $s_x>\alpha s_x>s_y>\alpha s_y$, hence
	\[
		\Delta_y\geq N_{yy}\bigl[s_x(\alpha-h_{xy})-\alpha s_y(\alpha-h_{xy})\bigr]=N_{yy}(\alpha-h_{xy})(s_x-\alpha s_y)>0.
	\]
	Case (iii) is symmetric.
\end{proof}

The trichotomy has a natural interpretation: $\alpha s_x$ and $\alpha s_y$ are the flows along the new
edge in the two directions, and the proof shows that if one of the endpoints receives from the other
endpoint more than $1/\alpha$ times what it gives back, then the receiving endpoint certainly gains
score; when the exchange along the new edge is balanced, both endpoints gain (except that, in the
degenerate boundary case described in (i), one of them might keep its score). Thus, an endpoint can
lose score only in cases (ii) and (iii), and in that case it is the endpoint that gives to the other
one more than $1/\alpha$ times what it receives. In all cases at least one endpoint gains score,
which provides an alternative proof of Theorem~\ref{thm:pr-score}. Note also that the theorem holds
for \emph{every} preference vector, and in particular also for vertices $x$ and $y$ with zero
preference, as long as one of them has nonzero score. Figure~\ref{fig:trichotomy} shows the three
cases in the $(s_x,s_y)$ plane for a few values of $\alpha$: the larger the damping factor, the
smaller the imbalance between the two flows that suffices to single out the endpoint that gains.

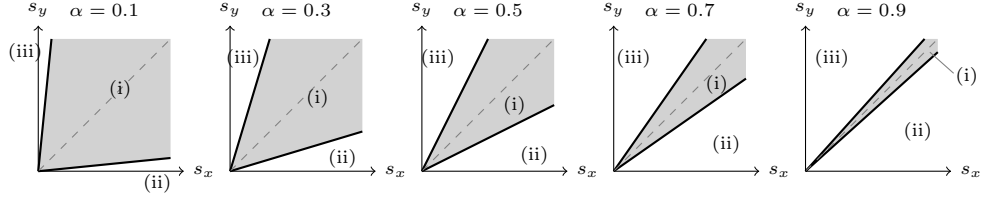
\begin{figure}
	\centering
	\begin{tikzpicture}[scale=1.75]
		\foreach \a/\dx/\iix/\iiy/\iiix/\iiiy/\ix/\iy in {0.1/0/0.9/-0.09/-0.1/0.9/0.62/0.62, 0.3/1.45/0.85/0.1/0.1/0.85/0.66/0.55, 0.5/2.9/0.85/0.12/0.12/0.85/0.7/0.5, 0.7/4.35/0.85/0.2/0.15/0.85/0.78/0.66, 0.9/5.8/0.85/0.3/0.2/0.85/1.22/0.72} {
			\begin{scope}[shift={(\dx,0)}]
				\fill[gray!35] (0,0) -- (1,\a) -- (1,1) -- (\a,1) -- cycle;
				\draw[thick] (0,0) -- (1,\a);
				\draw[thick] (0,0) -- (\a,1);
				\draw[dashed, gray] (0,0) -- (1,1);
				\draw[->] (0,0) -- (1.1,0) node[right, font=\scriptsize] {$s_x$};
				\draw[->] (0,0) -- (0,1.1) node[above, font=\scriptsize] {$s_y$};
				\node[font=\scriptsize] at (0.5,1.22) {$\alpha=\a$};
				\node[font=\scriptsize] at (\iix,\iiy) {(ii)};
				\node[font=\scriptsize] at (\iiix,\iiiy) {(iii)};
				\node[font=\scriptsize] at (\ix,\iy) {(i)};
			\end{scope}
		}
		\draw[gray, thin] (6.92,0.74) -- (6.74,0.9);
	\end{tikzpicture}
	\caption{\label{fig:trichotomy}The three cases of Theorem~\ref{thm:pr-trichotomy} in the
	$(s_x,s_y)$ plane, for five values of the damping factor $\alpha$. The two lines through the
	origin are $s_y=\alpha s_x$ and $s_y=s_x/\alpha$; the dashed line is $s_y=s_x$. In the shaded
	region (i) both endpoints of the new edge gain score; below the shallower line, in region (ii), $y$
	gains score; above the steeper line, in region (iii), $x$ gains score. As $\alpha$ grows, region
	(i) shrinks toward the diagonal.}
\end{figure}

Rank semi-monotonicity is an instance of Theorem~\ref{thm:damped}. In this case the correction term
of Lemma~\ref{lemma:damped-decomp} has a simple form: a last-passage walk from $x$ is a step to a
neighbor $w\neq x,y$ of $x$ followed by a walk from $w$ that does not touch $x$ or $y$, so
$\Psi_{xz}=\alpha\sum_wP_{xw}\nott N_{wz}$, and analogously for $y$; hence
$\gamma_z=s_x\Psi_{xz}+s_y\Psi_{yz}$, and
\begin{equation}
	\label{eq:pr-decomp}
	\Delta_z=\Psi_{xz}\bigl(\Delta_x-s_x\bigr)+\Psi_{yz}\bigl(\Delta_y-s_y\bigr).
\end{equation}

\begin{theorem}
	\label{thm:pr-rank}
	Let $G$ be an undirected graph without isolated vertices, $\alpha\in(0\..1)$, $\bm v$ a
	preference vector, and $x\noadj y$ be such that $p(x)>0$ and $p(y)>0$. Then for every $z\neq
	x,y$ either $p(z)=p'(z)=0$, or $p(z)>0$ and
	\[
		\frac{p'(z)}{p(z)}<\max\left(\frac{p'(x)}{p(x)},\frac{p'(y)}{p(y)}\right).
	\]
	Consequently, if all scores are nonzero, PageRank is strictly ratio semi-monotone, and thus
	strictly rank semi-monotone, on every undirected graph without isolated vertices, for every
	damping factor and every preference vector. More generally, if $e\in\{x,y\}$ attains the maximum
	above, then $p'(z)<p'(e)$ for all $z\neq x,y$ such that $p(z)\leq p(e)$.
\end{theorem}
\begin{proof}
	By Theorem~\ref{thm:pr-score}, $\max\bigl(p'(x)/p(x),p'(y)/p(y)\bigr)>1$, and
	$\gamma_z=s_x\Psi_{xz}+s_y\Psi_{yz}$ is positive unless $\Psi_{xz}=\Psi_{yz}=0$, as $s_x,s_y>0$: the
	condition for strict inequality of Theorem~\ref{thm:damped} holds for every $z$, and the
	statements follow from Theorem~\ref{thm:damped}, Corollary~\ref{cor:damped-rank} and
	Theorem~\ref{thm:ratio-semi-mon}.
\end{proof}

The condition $p(x),p(y)>0$ of Theorem~\ref{thm:pr-rank} cannot be weakened to the condition
$p(x)>0$ or $p(y)>0$ of Theorem~\ref{thm:pr-score}.

\begin{figure}
	\centering
	\begin{tikzpicture}[main/.style = {draw, circle, minimum size=5mm, inner sep=0pt}, scale=0.9]
		\node[main] (lone) at (0,1) {};
		\node[main] (ltwo) at (0,-1) {};
		\node[main] (c) at (1,0) {$c$};
		\node[main] (y) at (3,0) {$y$};
		\node[main] (x) at (5,0) {$x$};
		\node[main] (w) at (7,0) {$w$};
		\draw (lone) -- (c) -- (y);
		\draw (ltwo) -- (c);
		\draw (x) -- (w);
		\draw[dashed] (x) -- (y);
	\end{tikzpicture}
	\caption{\label{fig:pxy}A graph used to show that the condition $p(x),p(y)>0$ of Theorem~\ref{thm:pr-rank} cannot be weakened.}
\end{figure}
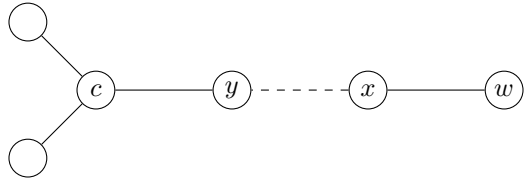

To show this, consider the graph $G$ of Figure~\ref{fig:pxy}, and let $\bm v$ be concentrated on $x$. Then $p(x)=1/(1+\alpha)$, $p(w)=\alpha/(1+\alpha)$, and all other scores are
zero, whereas $G'=G+x\adj y$ is connected, and a direct computation gives
\[
	p'(x)=\frac{2(6-5\alpha^2)}{(1+\alpha)(12-7\alpha^2)},\quad
	p'(y)=\frac{2\alpha(3-2\alpha^2)}{(1+\alpha)(12-7\alpha^2)},\quad
	p'(c)=\frac{3\alpha^2}{(1+\alpha)(12-7\alpha^2)}.
\]
Thus, $p'(c)>p'(y)$ iff $4\alpha^2+3\alpha>6$, that is, $\alpha>(\sqrt{105}-3)/8\approx0.906$, and
$p'(c)>p'(x)$ iff $13\alpha^2>12$, that is, $\alpha>2\sqrt{39}/13\approx0.961$. For $\alpha=49/50$,
say, we have $p(c)=p(y)=0<p(x)$, but $p'(c)>p'(x)>p'(y)$ (the three values are approximately
$0.276$, $0.229$ and $0.202$): the center of the star overtakes both endpoints of the new edge, so
PageRank is not even (non-strictly) rank semi-monotone on $G$ with preference vector $\bm v$. Note
that $y$ gains score, as guaranteed by Theorem~\ref{thm:pr-trichotomy}(ii). The same happens for
every preference vector supported on $\{x,w\}$, since replacing $\bm v$ with $v_x\bm e_x+v_w\bm e_w$
multiplies the scores of $x$, $y$ and $c$ in $G'$ by the same factor $v_x+\alpha v_w$.

We remark that (\ref{eq:reversibility}) holds also for $\nott N$, as $\nott N=(1-\alpha
P_S)^{-1}=(D_S-\alpha A_S)^{-1}D_S$, where $S=V_G\setminus\{x,y\}$ and $P_S$, $D_S$, $A_S$ are the
submatrices of $P$, $D$, $A$ indexed by $S$; since $\Phi_{zx}=\alpha\sum_w\nott N_{zw}P_{wx}$, we
obtain $\Psi_{xz}/d(z)=\Phi_{zx}/d(x)$, and~(\ref{eq:pr-decomp}) becomes
\[
	\frac{\Delta_z}{d(z)}=\Phi_{zx}\,\frac{\Delta_x-s_x}{d(x)}+\Phi_{zy}\,\frac{\Delta_y-s_y}{d(y)}.
\]
This shows that the increase of PageRank \emph{per unit of degree} (which is, up to a factor
$1-\alpha$, the increase of the PageRank of the walker seen from the point of view of an edge) is
always strictly maximized at $x$ or $y$; the rank semi-monotonicity statement is instead about
PageRank itself, and this is why we have to pass through ratios rather than differences.

\section{Impact of Loops}
\label{sec:loops}

Our definition of undirected graph excludes loops. In this section we discuss what happens if loops
are allowed,: a loop at $u$ is a further neighbor of $u$, so that
$A_{uu}=1$ and the loop contributes one to the degree $d(u)$; this is the convention under which the
uniform random walk moves from $u$ to each of its neighbors, $u$ included, with the same
probability. Other conventions are possible and natural in other contexts (for instance, a loop can
be seen as two arcs of a symmetric directed graph, contributing two to the degree; see the
discussion in~\cite{fibrations}), and would lead to slightly different statements. The new edge
$x\adj y$ is always assumed not to be a loop, that is, $x\neq y$.

A vertex whose only neighbor is itself behaves, for the purposes of this paper, like an isolated
vertex; accordingly, in the results that assume the absence of isolated vertices, the hypothesis
must be replaced by the hypothesis that \emph{every vertex has a neighbor different from itself}.

\paragraph{Katz's index.} Nothing changes. The walk-counting framework of Section~\ref{sec:walks} and
the comparison theorem of Section~\ref{sec:damped} hold for arbitrary nonnegative matrices, adding
the edge $x\adj y$ is still an admissible update of $A$, and the proof of Theorem~\ref{thm:katz}
uses only that $N_{xx}\geq1$ and $\nott N_{zz}\geq1$, which is true because of the empty walk. Thus
Theorem~\ref{thm:katz} holds verbatim on undirected graphs with loops, and so does the example
following it (a loopless graph) showing that Katz's index is not $\delta$-semi-monotone.

\paragraph{Eigenvector centrality.} Again, nothing changes. The Perron--Frobenius theorem requires
only that $A$ be irreducible, that is, that $G$ be connected (loops even make $A$ primitive), and
$\lambda'>\lambda$ still follows from $A'\geq A$, $A'\neq A$. Lemma~\ref{lemma:eig-identity},
Theorem~\ref{thm:eigenvector} and Corollary~\ref{cor:eig-score} never use the absence of loops.

\paragraph{PageRank.} Here the situation is more delicate, because the absence of loops is used in
the proofs of Section~\ref{sec:pagerank} in three places: to find a neighbor of $x$ or $y$
different from $x$ and $y$ (in the proofs of Theorems~\ref{thm:pr-score}
and~\ref{thm:pr-trichotomy}(i)), which is again guaranteed if every vertex has a neighbor different
from itself; in the bound $f_y\leq\alpha^2$ of the proof of Theorem~\ref{thm:pr-trichotomy}(ii),
which fails if there is a loop at $y$; and in the bound $f_x\leq\alpha^2$ of the proof of
Theorem~\ref{thm:pr-trichotomy}(i), which fails if there is a loop at $x$. Note that if both $x$ and
$y$ have themselves as their only neighbor and $v_x=v_y$, then $\bm p'=\bm p$, so the hypothesis
that every vertex has a neighbor different from itself is necessary.

Under that hypothesis, Theorem~\ref{thm:pr-score} and Theorem~\ref{thm:pr-rank} hold with the same
proofs, as they never use the bounds on $f_x$ and $f_y$. Cases (ii) and (iii) of
Theorem~\ref{thm:pr-trichotomy} also hold. Indeed, the bound $f_y\leq\alpha^2$ is used only to show
that $s_x(\alpha-h_{xy})+s_y(\alpha h_{xy}-f_y)>0$ when $s_y<\alpha s_x$, and for this purpose the
inequality $\alpha(1-f_y)\geq h_{xy}(1-\alpha^2)$ is sufficient: if $f_y\leq\alpha h_{xy}$ there is
nothing to prove, and otherwise
\[
	s_x(\alpha-h_{xy})+s_y(\alpha h_{xy}-f_y)>s_x(\alpha-h_{xy})-\alpha s_x(f_y-\alpha h_{xy})=s_x\bigl[\alpha(1-f_y)-h_{xy}(1-\alpha^2)\bigr]\geq0.
\]
We now show that the inequality holds also if there is a loop at $y$. We can assume that $h_{xy}>0$,
for otherwise $N_{xy}=0$ and $\Delta_y=s_y+(\alpha s_x-s_y)N_{yy}>0$. Let $S$ be the set of vertices
different from $y$ having a neighbor different from $y$, and $H=\max_{z\in S}h_{zy}$. A vertex of
$S$ adjacent to $y$ has degree at least two, and all neighbors of a vertex of $S$ different from
$y$ belong to $S$; thus, if $z\in S$ attains the maximum,
\[
	H=h_{zy}=\alpha\Bigl[P_{zy}+\sum_{v\neq y}P_{zv}h_{vy}\Bigr]\leq\alpha\bigl[P_{zy}+(1-P_{zy})H\bigr]\leq\frac{\alpha(1+H)}2,
\]
that is, $H\leq\alpha/(2-\alpha)$. Since $x$ and its neighbors belong to $S$, we have
$h_{xy}\leq\alpha H\leq\alpha^2/(2-\alpha)$. Moreover, the last vertex before $y$ of a walk from
$x$ to $y$ is a neighbor of $y$ belonging to $S$, whereas the remaining $d(y)-2$ neighbors $z\neq y$
of $y$ satisfy $h_{zy}\leq\alpha$; hence
\[
	f_y=\frac{\alpha}{d(y)}\Bigl[1+\sum_{z\adj y,\,z\neq y}h_{zy}\Bigr]\leq\frac{\alpha}{d(y)}\Bigl[1+(d(y)-2)\alpha+\frac{\alpha}{2-\alpha}\Bigr],
\]
and a direct computation gives
\[
	\alpha(1-f_y)-h_{xy}(1-\alpha^2)\geq\frac{2\alpha(1-\alpha)^2\bigl(d(y)(1+\alpha)-\alpha\bigr)}{d(y)(2-\alpha)}>0.
\]

Case (i) of Theorem~\ref{thm:pr-trichotomy}, instead, survives only in part: no vertex $z\neq x,y$
gains score, at least one of them loses score, and $\Delta_x+\Delta_y>0$, with the same proof; but
the endpoint with a loop may lose score, and the characterization of the equality case in terms of
pendant neighbors (which rests on $f_x\leq\alpha^2$) is lost. Let $G$ be the graph
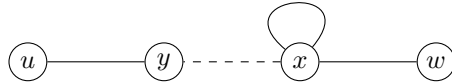
\[
	\begin{tikzpicture}[main/.style = {draw, circle, minimum size=5mm, inner sep=0pt}, scale=0.9]
		\node[main] (u) at (0,0) {$u$};
		\node[main] (y) at (2,0) {$y$};
		\node[main] (x) at (4,0) {$x$};
		\node[main] (w) at (6,0) {$w$};
		\draw (u) -- (y);
		\draw (w) -- (x);
		\draw[dashed] (x) -- (y);
		\draw (x) to[loop above, out=50, in=140, looseness=8] (x);
	\end{tikzpicture}
\]
and let $\bm v=\frac34\bm e_w+\frac14\bm
e_u$. Then $p(x)=3\alpha/(2(\alpha+2))$ and $p(y)=\alpha/(4(\alpha+1))$ are positive, and a direct
computation gives, writing $E=6+4\alpha-2\alpha^2-\alpha^3$ (which is positive on $(0\..1)$),
\begin{align*}
	s_y-\alpha s_x & =\frac{\alpha(1-\alpha)(3-\alpha-3\alpha^2)}{4E}, \\
	s_x-\alpha s_y & =\frac{\alpha(1-\alpha)(6+4\alpha-\alpha^2)}{4E},   \\
	\Delta_x       & =-\frac{3\alpha^3(\alpha+1)}{4(\alpha+2)E}.
\end{align*}
Thus, for $\alpha\leq(\sqrt{37}-1)/6\approx0.847$ we are in case (i), but $x$ loses score for every
$\alpha\in(0\..1)$ (for instance, for $\alpha=1/2$ the score of $x$ decreases from $3/10$ to
$69/236$). The reason is that the loop makes the killed walk started at $x$ return to $x$ in a single
step with probability $\alpha/2$, so that $f_x=\alpha(1+\alpha)/2>\alpha^2$. As a consequence, the
statement that an endpoint can lose score only in cases (ii) and (iii) depends on the absence of
loops, too.

\paragraph{Seeley's index.} Finally, we consider Seeley's index, which we did not discuss because
on loopless graphs its properties follow immediately from~\cite{boldi_furia_vigna_2023}. On a graph
with loops, the stationary distribution of the uniform random walk is $d(u)/S$, where
$S=\sum_vd(v)$ is the sum of all degrees (which is no longer twice the number of edges), and we
take this as the definition of Seeley's index. Then Seeley's index is still score monotone and
strictly rank monotone, with the same proof as in~\cite{boldi_furia_vigna_2023}, provided that
every vertex has a neighbor different from itself. Indeed, for $z\neq x,y$ we still have
$c'(z)-c(z)=d(z)/(S+2)-d(z)/S<0$, and $c'(z)<c'(x)$ if $d(z)\leq d(x)$, whereas
\[
	c'(x)-c(x)=\frac{d(x)+1}{S+2}-\frac{d(x)}{S}=\frac{S-2d(x)}{S(S+2)},
\]
so we need $S>2d(x)$, that is, $\sum_{v\neq x}d(v)>d(x)$. Every neighbor $v\neq x$ of $x$
contributes at least one to the left-hand side, for a total of at least $d(x)-A_{xx}$; moreover,
$y$ is not adjacent to $x$ and has a neighbor $w\neq y$ (necessarily $w\neq x$), and the edge
$y\adj w$ contributes at least two more (one from $d(y)$, and one from $d(w)$, since if $w$ is a
neighbor of $x$ its degree is at least two). Hence $\sum_{v\neq x}d(v)-d(x)\geq2-A_{xx}\geq1$. The
hypothesis cannot be dropped: if $x$ has a loop and a single further neighbor, which is a pendant
vertex, and $y$ has just a loop, then $S=4$, $d(x)=2$ and $c'(x)=c(x)=1/2$; and if both $x$ and $y$
have just a loop, neither of them changes its score.

\section{Conclusions}
\label{sec:conclusions}

\begin{table}
	\centering
	{
		\begin{tabular}{l|ll||c|c||c|c}
			\multicolumn{3}{c||}{}                                                             & \multicolumn{2}{c||}{monotonicity~\cite{boldi_furia_vigna_2023}} & \multicolumn{2}{c}{semi-monotonicity}                        \\
			\multicolumn{3}{c||}{}                                                             & score & rank   & score       & rank            \\\hline
			\multirow{6}{*}{Spectral} & \multirow{3}{*}{Eigenvector $\left\{\vrule height 1.5em depth 1.5em width 0pt\right.$} & $\ell_1$ & no & no & {\bf yes} & {\bf strict} \\
			                          &                                                        & $\ell_2$ & no & no & {\bf yes}   & {\bf strict}    \\
			                          &                                                        & proj.    & no & no & {\bf no}    & {\bf strict}    \\
			                          & \multicolumn{2}{l||}{Katz's index}                     & yes   & no     & yes         & {\bf strict}    \\
			                          & \multicolumn{2}{l||}{PageRank}                         & no    & no     & {\bf yes}   & {\bf strict}    \\
			                          & \multicolumn{2}{l||}{Seeley's index}                   & yes   & strict & yes         & strict          \\\hline
			\multirow{3}{*}{Geometric~\cite{BDFSRSM}} & \multicolumn{2}{l||}{Closeness}        & yes   & no     & yes         & yes             \\
			                          & \multicolumn{2}{l||}{Harmonic centrality}              & yes   & no     & yes         & strict          \\
			                          & \multicolumn{2}{l||}{Betweenness}                      & no    & no     & no          & yes             \\
		\end{tabular}}
	\vspace*{3mm}
	\caption{\label{tab:summ}Summary of the results about spectral and geometric centralities on
	connected undirected graphs. Results in boldface are new; the results about geometric
	centralities are from~\cite{BDFSRSM}, and all results about monotonicity are
	from~\cite{boldi_furia_vigna_2023}. For rank semi-monotonicity, ``yes'' means that the property
	holds but not in its strict form, whereas ``strict'' means that the strict form holds. All results are valid also for graphs with loops, under the convention that a loop adds one to the degree, provided every vertex has a neighbor other than itself (a hypothesis that matters only for PageRank and Seeley's index).}
\end{table}

Table~\ref{tab:summ} summarizes our results. Semi-monotonicity, in its strict rank form, is
satisfied by all three spectral centralities, and score semi-monotonicity is satisfied by all of them
except for eigenvector centrality under a non-monotone normalization. In particular, the answer to
the question raised in~\cite{boldi_furia_vigna_2023}---can a new collaboration hurt the PageRank of
both collaborators?---is negative. Connectivity is essential only for eigenvector centrality: the
results about Katz's index hold on every undirected graph, and those about PageRank hold on every
graph without isolated vertices as
long as all scores are nonzero, which is the hypothesis of the rank-monotonicity results
of~\cite{boldi_furia_vigna_2023}; the example following Theorem~\ref{thm:pr-rank} shows that this
hypothesis cannot be dropped.

It is interesting to contrast these results with those about geometric centralities and
betweenness~\cite{BDFSRSM}: for the latter, semi-monotonicity is proved through basin dominance, an
additive property of shortest-path distances, and strictness fails for closeness and betweenness;
for spectral centralities, instead, the natural property is multiplicative, and strictness is always
attained. All our positive results are instances of a single comparison theorem for damped spectral
centralities~\cite{vigna} (Theorem~\ref{thm:comparison}), which rests on the fact that the score of a
vertex $z\neq x,y$ is a nonnegative combination of the scores of $x$ and $y$ plus a term accounting
for the walks that never pass through $x$ or $y$: for Katz's index and PageRank the theorem is
applied to the increment, which is itself a damped spectral centrality, and it holds for all updates
that modify the rows of the two endpoints and can only decrease entries outside of the corresponding
columns; for eigenvector centrality, it is applied to the new and the old eigenvector, both of which
are damped spectral centralities of the old graph. Ratio semi-monotonicity might be useful for other
measures defined through walks or random processes on the graph, and it would be interesting to know
whether it holds for further modifications of the graph (e.g., adding several edges at once).

\bibliography{bibliography}

\end{document}